\documentclass[12pt]{article}
\usepackage{color}
\usepackage{amssymb,latexsym,bm, amsmath,amsthm, amsfonts,multirow, graphicx,mathrsfs}
\usepackage{eufrak,booktabs}
\usepackage{epsfig}
\usepackage{enumitem}
\usepackage[hyphens]{url}
\usepackage{thmtools}

\newcommand{\argmax}{\mathop{\rm argmax}}

\newtheorem{thm}{Theorem}[section]
\newtheorem{prop}{Proposition}[section]
\newtheorem{defn}[thm]{Definition}

\newtheorem{rmk}[thm]{Remark}

\newtheorem{ex}{Example}[section]

\def\cov{{\rm{cov}}}

\def\tr{{\rm{tr}}}

\def\rmspan{{\rm{span}}}

\def\SIR{{\rm{SIR}}}
\def\PSIR{{\rm{PSIR}}}
\def\SAVE{{\rm{SAVE}}}

\def\real{{\mathbb R}}

\def\syx{{\mathcal{S}_{Y|X}}}

\def\syxw{{\mathcal{S}_{Y|X}^{(W)}}}

\def\syxe{{\mathcal{S}_{E[Y|X]}}}

\def\syz{{\mathcal{S}_{Y|Z}}}
\def\sywz{{\mathcal{S}_{(Y,W)|Z}}}
\def\swz{{\mathcal{S}_{W|Z}}}
\def\syze{{\mathcal{S}_{E[Y|Z]}}}

\def\syzw{{\mathcal{S}_{Y|Z}^{(W)}}}

\def\syzww{{\mathcal{S}_{Y_w|Z_w}}}
\def\se{{\mathcal{S}_{\rm{env}}}}

\def\de{{d_{\rm{env}}}}
\def\dehat{{\widehat{d}_{\rm{env}}}}
\def\Be{{B_{\rm{env}}}}
\def\Behat{{\widehat{B}_{\rm{env}}}}

\def\Kyz{{K_{Y|Z}}}
\def\Kywz{{K_{(Y,W)|Z}}}
\def\Kwz{{K_{W|Z}}}
\def\Kyzw{{K_{Y|Z}^{(W)}}}

\def\Kyzhat{{\widehat{K}_{Y|Z}}}
\def\Kywzhat{{\widehat{K}_{(Y,W)|Z}}}
\def\Kwzhat{{\widehat{K}_{W|Z}}}
\def\Kyzwhat{{\widehat{K}_{Y|Z}^{(W)}}}
\def\Kehat{{\widehat{K}_{\rm{env}}}}

\renewcommand\thmcontinues[1]{Continued}

\newcommand{\indep}{\;\, \rule[0em]{.03em}{.67em} \hspace{-.25em}
\rule[0em]{.65em}{.03em} \hspace{-.25em}
\rule[0em]{.03em}{.67em}\;\,}

\definecolor{DarkBlue}{rgb}{0,0.1,0.9}

\title{Sufficient dimension reduction with additional information}
\begin{document}

\author{ Hung Hung$^{a}$, Chih-Yen Liu$^{a}$, and Henry Horng-Shing Lu$^{b}$\\
\small $^{a}$Institute of Epidemiology and Preventive Medicine,
National Taiwan University, Taiwan\\
\small $^{b}$Institute of Statistics, National Chiao Tung
University, Taiwan}
\date{}
\maketitle

\begin{abstract}

Sufficient dimension reduction is widely applied to help model
building between the response $Y$ and covariate $X$. While the
target of interest is the relationship between $(Y,X)$, in some
applications we also collect additional variable $W$ that is
strongly correlated with $Y$. From a statistical point of view,
making inference about $(Y,X)$ without using $W$ will lose
efficiency. However, it is not trivial to incorporate the
information of $W$ to infer $(Y,X)$. In this article, we propose a
two-stage dimension reduction method for $(Y,X)$, that is able to
utilize the additional information from $W$. The main idea is to
confine the searching space, by constructing an envelope subspace
for the target of interest. In the analysis of breast cancer data,
the risk score constructed from the two-stage method can well
separate patients with different survival experiences. In the Pima
data, the two-stage method requires fewer components to infer the
diabetes status, while achieving higher classification
accuracy than conventional method.\\

\end{abstract}

\section{Introduction}\label{sec.introduction}

It is common to construct regression models or classification rules
based on the observed data. Among the collected covariates, it is
usually the case that some of them possess better performances than
the rest in inferring the response, but with higher obtaining cost.
Let $Y\in \mathbb{R}$ be the response of interest and $(X,W)$ be the
covariates, where $W\in \mathbb{R}$ is a better predictor of $Y$
than $X\in \mathbb{R}^p$. In the breast cancer data, for instance,
$Y$ is the survival time, $X$ contains 30 real-valued features from
digitized image of a fine needle aspirate, and $W$ is the index of
breast cancer stage which is defined by tumor size and the number of
lymph nodes. Obviously, the cancer stage provides more information
about the survival experience, but it is an invasive process to
collect $W$. Take the Pima data as another example, $Y$ represents
the diabetes status, $X$ is the vector of 7 biological measurements,
and $W$ is the score of the family disease history. It is
empirically shown that the family disease history is influential on
the diabetes status, but may not always be available due to its
frequent missingness. Consequently, although we can study the
relationship between $Y$ and $(X,W)$, the applicability of the
constructed model is limited since most of the collected subjects in
the future are not willing to (e.g., invasive test) or not able to
(e.g., missingness of the disease history) have the information of
$W$. In this situation, the research aim is to construct a
pre-screening model based on $(Y,X)$ only. For those susceptible
subjects, $W$ is collected (with more cost) to help identify the
true behavior of $Y$ (with higher precision) for further treatment.

A common situation in modern biomedical research is the
high-dimensionality of $X$, which makes the model building of
$(Y,X)$ difficult. Sufficient dimension reduction has been proposed
to reduce the dimension of $X$ while preserving its information for
$Y$ without requiring distributional assumption on $(Y,X)$. It aims
to search a matrix $\Gamma$ such that
\begin{eqnarray}
Y\indep X|\Gamma^TX. \label{edr}
\end{eqnarray}
It implies that all the information of $X$ with respect to $Y$ is
contained in $\Gamma^TX$. The space $\rmspan(\Gamma)$ is called the
dimension reduction subspace for the regression of $Y$ with respect
to $X$. Note that (\ref{edr}) always holds by taking $\Gamma$ as the
identity matrix $I_p$, but this choice is practically useless since
the dimension of $X$ is not reduced. The central subspace (CS)
induced by (\ref{edr}), denoted by $\syx$, is the intersection of
all such subspaces that satisfy (\ref{edr}), which carries least but
sufficient information of $X$  regarding $Y$. In this article, we
assume the existence and uniqueness of $\syx$, which can be
guaranteed under very mild conditions (Cook, 1998). Let
$d=\dim(\syx)$ be the structural dimension of $\syx$. The
relationship of $(Y,X)$ can be explored by the $(d+1)$-dimensional
plot of $(Y,\Gamma^TX)$, which is useful to model $(Y,X)$. Another
important concept related to this study is partial sufficient
dimension reduction (Chiaromonte, Cook, and Li, 2002), which aims to
find the intersection of all subspaces $\rmspan (\Gamma_{W})$ such
that
\begin{eqnarray}
Y\indep X|(\Gamma_{W}^{T}X, W). \label{partial_dr}
\end{eqnarray}
The resulting subspace is called the partial central subspace (PCS)
for the regression of $Y$ on $X$ given $W$, and is denoted by
$\syxw$.

Turning to the problem of constructing a model for $(Y,X)$, one can
directly apply any dimension reduction method on $(Y,X)$ to estimate
$\syx$. However, this simple strategy does not utilize the
information of $W$. As mentioned in our motivating examples, $W$
contains more information about $Y$ than $X$, and ignoring $W$ could
suffer the problem of inefficiency. This phenomenon can be partially
observed through the following example.
\begin{ex}\label{ex.model1}
Let $X\sim N(0,I_p)$ and $W|X \sim N(\beta^TX, 1-b^2)$ with
$b=\|\beta\|<1$. Let
\begin{eqnarray}
&& Y|(X,W)\sim N\left(\gamma^T X+aW,\sigma^2\right)\nonumber\\
&\Rightarrow& Y|X \sim N\left((\gamma+a\beta)^T X,
\sigma^2+a^2(1-b^2)\right),\label{model1}
\end{eqnarray}
where $\gamma\in \mathbb{R}^p$, $a\in \mathbb{R}^+$ controls the
influence of $W$ in explaining $Y$, $b$ controls the correlation
between $(X,W)$, and $\sigma^2$ is the conditional variance of $Y$
given $(X,W)$. It gives $\syx=\rmspan(a\beta+\gamma)$ and
$\syxw=\rmspan(\gamma)$
\end{ex}
\noindent One can observe from Example~\ref{ex.model1} that, without
considering $W$, the conditional variance of $Y$ is incremented by
$a^2(1-b^2)$, which is an increasing function of $a$. When $a$ is
large (i.e., $W$ plays an important role in affecting $Y$),
estimation procedure using $(Y,X,W)$ suffers less inherent
variation, and hence, has a chance of being more efficient than
using $(Y,X)$ only. Incorporating $W$ into the estimation of $\syx$
is not trivial. One possibility is to use the whole data $(Y,X,W)$
to estimate $\syxw$, and see if there are certain connections
between $\syxw$ and $\syx$. Chiaromonte, Cook, and Li (2002) show
that either $W \indep X | \Gamma_W^TX$ or $W \indep
Y|\Gamma_{W}^{T}X$ implies $\syx\subseteq\syxw$, and $W \indep Y|X$
implies $\syxw\subseteq\syx$. However, the stated conditions are not
easy to check in practice, and $\syx\neq\syxw$ in general. In
Example~\ref{ex.model1}, for instance, $\syxw=\rmspan(\gamma)$ is
different from the target of interest $\syx=\rmspan(a\beta+\gamma)$.
The aim of this study is to propose a dimension reduction method
that targets $\syx$ correctly, while utilizing all the information
of $(Y,X,W)$.

The rest of this article is organized as follows. In Section 2 we
review some dimension reduction methods. A two-stage estimation
procedure for $\syx$ that utilizes $W$ is introduced in Section 3.
Sections 4-5 conduct numerical studies to show the superiority of
the proposed method. The paper is ended with a discussion in
Section~6.

\section{Reviews of Dimension Reduction Methods}\label{sec.review}

\subsection{Preliminary}

In the rest of discussion, $Y$ is assumed to be a discrete random
variable with finite support $(1,\ldots,H)$. For the continuous
case, the discretization procedure of Li (1991) can be applied. For
the purpose of illustration, $W$ is also assumed to be a discrete
random variable with finite support $(1,\ldots,C)$. Extensions to
general $W$ will be discussed. Let $Z=\Sigma^{-1/2}(X-\mu)$ with
$\mu=E[X]$ and $\Sigma=\cov(X)$ be the standardized version of $X$.
There is no difference in considering the dimension reduction
problem for $X$ and $Z$, due to the relationships
$\syx=\Sigma^{-1/2}\syz$ and $\syxw=\Sigma^{-1/2}\syzw$. For
convenience, we will work on the $Z$-scale to introduce our method
and transform back to the $X$-scale. In practice, $Z$ is replaced by
$\widehat Z=\widehat\Sigma^{-1/2}(X-\widehat\mu)$, where
$\widehat\mu$ and $\widehat\Sigma$ are moment estimators of $\mu$
and $\Sigma$. Let $\{(Y_i,X_i,W_i)\}_{i=1}^n$ be random copies of
$(Y,X,W)$. Let $P_A$ be the orthogonal projection matrix onto a
space $A$, or $\rmspan(A)$ when $A$ stands for a matrix, and
$Q_A=I-P_A$. $I(\cdot)$ is the indicator function. For a population
quantity $\theta$, $\widehat\theta$ or $\widetilde\theta$ denotes
its sample version. In the subsequent discussion, we assume that all
the structural dimensions (e.g., $d$) are known first, and their
selections will be discussed separately.

\subsection{Estimation of $\syz$}

One branch of dimension reduction methods is to search a kernel
matrix $\Kyz$ satisfying $\rmspan (\Kyz)=\syz$. The solutions from
the maximization problem
\begin{eqnarray}
\max_{\substack{\beta_{s}:~\|\beta_{s}\|=1
\\\beta_{s}^{T}\beta_{l}=0\: \forall\: s \neq l}} \sum_{k=1}^{d}
\beta_{k}^T \Kyzhat \beta_{k} \label{criterion.1}
\end{eqnarray}
are used to estimate a basis of $\syz$. Inverse regression-based
methods usually rely on the linearity condition $E[Z|A^{T}Z]=P_{A}Z$
for first-order methods, such as sliced inverse regression (SIR) of
Li (1991), and further require the constant variance condition
$\cov(Z|A^{T}Z)=Q_{A}$ for second-order methods, such as sliced
average variance estimates (SAVE) of Cook and Weisberg (1991).

Undoubtedly, SIR is the most widely applied dimension reduction
method. The population kernel matrix of SIR is
$K_{\SIR}=\cov(E[Z|Y])$. Its sample version is
\begin{eqnarray}
\widehat K_{\SIR} = \sum\limits_{h=1}^{H}\frac{n_h}{n} m_h m_h^T,
\label{k.sir}
\end{eqnarray}
where
$m_h=\frac{1}{n_{h}}\sum\limits_{i=1}^{n}\widehat{Z}_{i}I(Y_{i}= h)$
is the slice mean and $n_h=\sum_{i=1}^nI(Y_i=h)$ is the size of the
slice $h$. Under the linearity condition, Li (1991) shows that the
span of the leading $d$ eigenvectors of $\widehat{K}_{\SIR}$ is a
$\sqrt{n}$-consistent estimator of $\syz$. SIR has the drawback of
not being able to identify $\syz$ when $E[Z|Y]$ is degenerate, and
SAVE is proposed to solve this problem (but requires both the
linearity and the constant variance conditions). The kernel matrix
of SAVE is given by $K_{\SAVE}= E\{I-\cov(Z|Y)\}^2$. At the sample
level, the leading eigenvectors of
\begin{eqnarray}
\widehat{K}_{\SAVE}=\sum\limits_{h=1}^{H}  \frac{n_h}{n} \left\{ I -
\widehat{\cov}(Z|Y= h) \right\}^2\label{k.save}
\end{eqnarray}
is used to estimate $\syz$, where $\widehat{\cov}(Z|Y=h)$ is the
sample covariance matrix of $\widehat Z_i$ within the slice $h$.

\subsection{Estimation of $\syzw$}
When the target of interest is $\syzw$, Chiaromonte, Cook, and Li
(2002) propose the partial sliced inverse regression (PSIR). Let
$(Y_w,Z_w)$ denote the random variables $(Y,Z)$ given $W=w$, and let
$Z_w^*=\Sigma_w^{-1/2}(Z_w-\mu_w)$ with $\mu_w=E[Z_w]$ and
$\Sigma_w=\cov(Z_w)$ being the standardized version of $Z_w$ under
$\{W=w\}$. The rationale of PSIR can be seen from the decomposition
\begin{eqnarray}\label{equ.psir}
\syzw=\bigoplus_{w=1}^{C} \syzww= \bigoplus_{w=1}^{C}
\Sigma_w^{-1/2}\mathcal{S}_{Y_w| Z_w^*}=\Sigma_0^{-1/2}
\bigoplus_{w=1}^{C} \mathcal{S}_{Y_w| Z_w^*},
\end{eqnarray}
where the first equality is from Proposition 3.3 of Chiaromonte,
Cook, and Li (2002) and the last equality holds under the equal
covariance condition $\Sigma_{w}=\Sigma_0, w=1,\ldots,C$. Let
$\widehat{\mu}_{w}$ and $\widehat{\Sigma}_{w}$ be the moment
estimators of $\mu_w$ and $\Sigma_w$, $\Sigma_0$ is estimated by
$\widehat{\Sigma}_0=\sum_{w=1}^{C}\frac{n_{w}}{n}
\widehat{\Sigma}_w$, where $n_w$ is the number of samples within
$\{W=w\}$. Let
\begin{eqnarray} \widehat
K_{\PSIR}=\sum_{w=1}^C\frac{n_w}{n} \widehat K_{w}^*, \label{k.psir}
\end{eqnarray}
where $\widehat K_{w}^{*}$ is the kernel matrix of SIR based on
$(Y_{w},\widehat{Z}_{w}^{*})$, and $\widehat
Z_w^*=\widehat\Sigma_0^{-1/2}(\widehat Z_w-\widehat\mu_w)$. A basis
of $\syzw$ can be estimated by $\widehat\Sigma_0^{-1/2}$ multiplying
the leading eigenvectors of $\widehat K_{\PSIR}$.


\section{Estimation of $\syz$ with Additional Information}\label{sec.two-stage}

\subsection{The $W$-envelope subspace and a two-stage method}\label{sec.est}

As mentioned in Section~\ref{sec.introduction} that $W$ contains
useful information regarding $Y$, and our aim is to incorporate $W$
into the estimation procedure of $\syz$. The basic idea is to use
$(Y,Z,W)$ to construct an \emph{envelope} that encapsulates the
searching space of $\syz$. With the confined searching space, we
have a chance to improve efficiency. The construction of such an
envelope is based on the fact that
\begin{eqnarray}
\syz\subseteq\sywz,
\end{eqnarray}
where the equality holds since $Y$ is a function of $(Y,W)$. The
inclusion property provides a way to utilize the information of $W$,
via constructing the $W$-{\it envelope subspace}.
\begin{defn}\label{def.se}
The $W$-envelope subspace of $\syz$ is defined to be $\se = \sywz$
with the structural dimension $\de=\dim(\se)$.
\end{defn}
\noindent Again, we assume that $\de$ is known and its selection
will be discussed later. Another expression of $\se$ via the concept
of PCS is established below.

\begin{prop}\label{prop.se}
$\se =\syzw \oplus \swz$.
\end{prop}

\noindent Although two expressions of $\se$ are equivalent in the
population level, we will see that the expression $\se =\syzw \oplus
\swz$ provides a more robust manner to construct $\se$. Since
$\sywz$ must exist, we always have the inclusion relationship
\begin{eqnarray}
\syz \subseteq \se. \label{sss}
\end{eqnarray}
Take Example~\ref{ex.model1} to exemplify, where
$\syz=\rmspan(a\beta+\gamma)$, $\swz=\rmspan(\beta)$, and
$\syzw=\rmspan(\gamma)$ due to $\cov(X)=I_p$. It can be seen that
$\syz$ is a proper subspace of $\se=\rmspan([\beta, \gamma])$.

Reasonably, it suffices to search $\syz$ within $\se$ due to
(\ref{sss}). An improved estimation procedure is to search a basis
of $\syz$ via solving the maximization problem
\begin{eqnarray}
\max_{\substack{\beta_{s}:~\beta_{s}\in\se\\
\|\beta_{s}\|=1,\beta_{s}^{T}\beta_{l}=0\: \forall\: s\neq l}}
\sum_{k=1}^{d}  \beta_{k}^T\Kyzhat \beta_{k} \label{criterion.2}.
\end{eqnarray}
Different from (\ref{criterion.1}), the estimation criterion
(\ref{criterion.2}) incorporates the information of $W$ via adding
the constraints $\beta_{s}\in\se$. Let $\Be$ be a basis of $\se$.
From Proposition~3 of Naik and Tsai (2005), the solutions of
(\ref{criterion.2}) are derived to be the leading $d$ eigenvectors
of
\begin{eqnarray}\label{ker_2}
P_{\Be} \Kyzhat P_{\Be}.
\end{eqnarray}
Observe from (\ref{ker_2}) that, instead of searching a basis of
$\syz$ in $\real ^ p$, we first project $\Kyzhat$ onto $\se$ within
which we search a basis of $\syz$. See Figure~\ref{env.plot} for a
conceptual display. Since $\se$ is rarely known a priori, we need to
estimate $P_{\Be}$ before estimating $\syz$ from (\ref{ker_2}). Let
$K_{\rm env}$ be a positive semi-definite kernel matrix satisfying
$\rmspan(K_{\rm env})=\se$, and let $\Behat$ be the leading $\nu$
($\nu\ge\de$) eigenvectors of $\widehat K_{\rm env}$. The two-stage
estimation procedure for $\syz$ is proposed to be the leading $d$
eigenvectors of
\begin{eqnarray}\label{ker_2_est}
P_{\Behat} \Kyzhat P_{\Behat}.
\end{eqnarray}

\noindent Obviously, the construction of $\widehat K_{\rm env}$
based on $(Y,Z,W)$ plays the key role to the performance of the
two-stage method, wherein the information of $W$ should be properly
utilized. This issue will be discussed in Section~\ref{sec.se}.

%
%
%
%
%




\subsection{The construction of $\Kehat$}\label{sec.se}

We now proceed to the construction of $\Kehat$, the critical part of
the two-stage method. One approach is to use the expression
$\se=\mathcal{S}_{(Y,W)|Z}$ directly. Let $\Kywz$ be a positive
semi-definite kernel matrix satisfying
$\rmspan(\Kywz)=\mathcal{S}_{(Y,W)|Z}$. Then, one can use
$\Kehat=\Kywzhat$ in the construction of the two-stage method, and
any existing dimension reduction method can be applied to obtain
$\Kywzhat$. For example, $\Kywz$ can be chosen to be the SIR kernel
matrix $K_{\rm SIR}^*=\cov(E[Z|Y,W])$. In the sample level, one can
still use (\ref{k.sir}) to construct $\widehat K_{\rm SIR}^*$,
except $(Y,W)$ are now treated as the response to do slicing. A
naive two-stage estimator for $\syz$ is proposed below.

\begin{defn}
The two-stage estimator of $\syz$ using the leading $\nu$
eigenvectors of $\Kehat=\Kywzhat$ is defined to be $\widehat
B_0(\nu)$.
\end{defn}


\noindent The efficiency gain of $\widehat B_0(\nu)$ is guaranteed
when the dimension is correctly specified at $\nu=\de$, which
supports the superiority of the two-stage method. Let $\widetilde B$
be the direct estimator of $\syz$ based on $(Y,Z)$. We have the
following result.

\begin{prop}\label{prop.efficiency}
Assume the linearity and constant variance conditions. Consider
$\widehat K_{Y|Z}=\widehat K_{\rm SIR}$. Then, $\widehat B_0(\de)$
with $\Kywzhat=\widehat K_{\rm SIR}^*$ is asymptotically more
efficient than $\widetilde B$ in estimating $\syz$, provided that
$\rmspan(K_{\rm SIR})\bigcap\rmspan(K_{\rm SIR}^*-K_{\rm
SIR})\ne\{0\}$.
\end{prop}

\begin{proof}
By treating $(Y,W)$ as the response and considering the function
$g(Y,W)=Y$, the result is a direct consequence of Theorem~1 of Hung
(2012).
\end{proof}


The naive two-stage estimator $\widehat B_0(\nu)$ can be improved.
One advantage of using $\Kehat=\Kywzhat$ is its simple
implementation, which can be conducted by existing algorithms with
slight modification. However, this method does not consider the
relative importance of $\syzw$ and $\swz$ when forming $\se$. For
instance, in the case of $W\indep Z$, $\swz=\{0\}$ and is useless to
improve estimating $\syz$. Using $\Kehat=\Kywzhat$ cannot adapt to
this situation and, hence, may loss efficiency in estimating $\se$
and $\syz$. The problem can be solved by constructing $\Kehat$ via
the alternative expression $\se=\syzw\oplus\swz$. Let $\Kyzw$ and
$\Kwz$ be two positive semi-definite kernel matrices satisfying
$\rmspan(\Kyzw)=\syzw$ and $\rmspan(\Kwz)=\swz$. From
Proposition~\ref{prop.se}, we have for any $\xi\in (0,1)$ that
\begin{eqnarray}\label{k.se}
\se=\rmspan \left( \xi\cdot \Kwz +(1-\xi)\cdot \Kyzw \right).
\end{eqnarray}
A more robust method is to construct $\Kehat$ by the hybrid kernel
matrix
\begin{eqnarray} \widehat K(\xi)=\xi \cdot \Kwzhat +(1-\xi) \cdot
\Kyzwhat ,\label{k.env}
\end{eqnarray}
where $\xi$ controls the relative importance of $\Kwzhat$ and
$\Kyzwhat$ in estimating $\se$. Observe that the construction of
$\Kwzhat$ is still a dimension reduction problem for $Z$, where $W$
is now treated as a response. Hence, we can apply SIR to construct
$\Kwzhat$ as in $(\ref{k.sir})$ with $Y$ being replaced by $W$. As
to $\Kyzwhat$, it is nothing but a partial dimension reduction
problem, and PSIR can be applied. An improved two-stage estimator
for $\syz$ is proposed below.

\begin{defn}
The two-stage estimator of $\syz$ using the leading $\nu$
eigenvectors of $\Kehat=\widehat K(\xi)$ is defined to be $\widehat
B(\nu,\xi)$.
\end{defn}

\noindent We note that any dimension reduction method can be applied
to construct $\Kehat$, and is not limited to SIR and PSIR.

\begin{rmk}
For multivariate $W$, one can still apply PSIR to construct
$\Kyzwhat$ by using $W$ to do slicing. As to the construction of
$\Kywzhat$ or $\Kwzhat$, this is a dimension reduction problem with
multivariate response, and the projective resampling technique (Li,
Wen, and Zhu, 2008) can be applied. With these modifications, the
same two-stage procedure is ready to estimate $\syz$ by using the
modified $\Kehat$.
\end{rmk}


\subsection{Determination of $(d, \de)$}\label{sec.eta}


To estimate the structural dimension, Li (1991) proposes an
asymptotic test based on the sum of the tail eigenvalues of the
kernel matrix. This method, however, requires the asymptotic
distribution of eigenvalues, which is complicated when $Z$ is not
normally distributed. Cook and Yin (2001) suggest a permutation test
to determine the structural dimension, but it requires heavy
computational load. Alternatively, we adopt a BIC-type criterion to
select $(d,\de)$, which is modified from the criterion of Zhu {\it
et al.} (2010). Let
\begin{eqnarray}\label{est.de_eta}
\dehat (\xi)=\argmax_{k=1,\ldots,p}
\left\{\frac{\sum_{j=1}^k\{\ln(\widehat\lambda_j+1)-\widehat\lambda_j\}}
{\sum_{j=1}^p\{\ln(\widehat\lambda_j+1)-\widehat\lambda_j\}}-
\frac{C_n}{n}(pk-\frac{k(k-1)}{2})\right\},
\end{eqnarray}
where $\widehat\lambda_j$ is the $j$-th eigenvalue of $\widehat
K_{\textrm{env}}$, $C_n$ is a pre-determined penalty, and
$pk-\frac{k(k-1)}{2}$ is the number of parameters required to
specify a $p\times p$ symmetric matrix with rank $k$. Note that
$\dehat(\xi)$ can be a function of $\xi$, depending on the choices
of $\Kehat=\widehat K(\xi)$ or $\Kehat=\Kywzhat$. To integrate out
the effect of $\xi$, we propose to estimate $\de$ by
\begin{eqnarray}\label{est.de0}
\dehat = {\rm median}\left\{\dehat(\xi):\xi\in\Xi \right\}.
\end{eqnarray}
The same idea can be applied to determine $d$. Let
\begin{eqnarray}
\widehat d(\xi)=\argmax_{k=1,\ldots,\dehat(\xi)}\left
\{\frac{\sum_{j=1}^k\{\ln(\widehat\lambda_j^*+1)-\widehat\lambda_j^*\}}
{\sum_{j=1}^p
\{\ln(\widehat\lambda_j^*+1)-\widehat\lambda_j^*\}}-\frac{C_n}{n}(pk-\frac{k(k-1)}{2})\right\},
\end{eqnarray}
where $\widehat\lambda_j^*$ is the $j$-th eigenvalue of $P_{\Behat}
\Kyzhat P_{\Behat}$ with $\nu=\widehat d_{\rm env}(\xi)$. We then
propose to estimate $d$ by
\begin{eqnarray}
\widehat d = \textrm{median}\left\{\widehat
d(\xi):\xi\in\Xi\right\}.
\end{eqnarray}

For any fixed $\xi$, the consistency of $\dehat (\xi)$ and $\widehat
d(\xi)$ can be similarly derived as Theorem~4 of Zhu {\it et al.}
(2010), provided $C_n/n \rightarrow 0$ and $C_n \rightarrow \infty$
as $n \rightarrow \infty$. Since $\Xi$ is finite, the consistency of
$(\widehat d, \dehat)$ is a direct consequence.






\section{Numerical Studies}\label{sec.simulation}

\subsection{Simulation settings}

We consider two models for simulations. The first one is model
(\ref{model1}) in Example~\ref{ex.model1} with
$\beta=b\cdot(0,0,1,1,0_{p-4})^T/\sqrt{2}$ and
$\gamma=(1,1,0,0,0_{p-4})^T/\sqrt{2}$, which gives
$\syz=\rmspan(a\beta +\gamma)$ and $\se=\rmspan([\beta,\gamma])$.
The second model is constructed as below. Let $X\sim N
\left(0,I_p\right)$  and $W= \left(W_{1},W_{2}\right)$ be generated
from $W_{1}|X \sim N\left(\beta_{1}^T X,1-\|\beta_{1}\|^2\right)$
and $W_{2}|X \sim N\left(\beta_{2}^T X,1-\|\beta_{2}\|^2\right)$,
where $\beta_{1}=b\cdot(0,0,1,1,0_{p-4})^T/\sqrt{2}$ and
$\beta_{2}=b\cdot(1,1,0,0,0_{p-4})^T/\sqrt{2}$. Condition on
$(X,W)$, $Y$ is generated from
\begin{eqnarray}
&&Y|(X,W) \sim N\left( (1+\alpha^{T} X)(aW_1+aW_2+\gamma^{T}
X),\sigma^{2}\right)\nonumber\\
&\Rightarrow& Y|X \sim N\left((1+\alpha^{T} X)(a\beta_1+ a\beta_2
+\gamma)^TX,\sigma^2+ 2a^2(1-b^{2})(1+\alpha^{T} X)^2 \right)
\label{model2}
\end{eqnarray}
with $\alpha = (0,0,0,0,1,1,0_{p-6}^T)^T/\sqrt{2}$ and $\gamma =
(1,1,2,2,0_{p-4}^T)^T/\sqrt{10}$. The setting is designed so that
$\gamma\in \rmspan([\beta_1,\beta_2])$, i.e.,
$\syzw=\rmspan([\alpha,\gamma])$ and
$\swz=\rmspan([\beta_1,\beta_2])$ have overlap. It gives
$\syz=\rmspan([\alpha, a(\beta_{1}+ \beta_{2}) +\gamma])$ and
$\se=\rmspan([\alpha,\beta_{1},\beta_{2}])$. Note that in both
models, $a$ controls the ability of $W$ to explain $Y$, and $b$
controls the correlation between $W$ and $Z$.

Simulation data is generated from two settings of $(n,p)=(150,9)$
and $(250,25)$. Both $\widehat B$ and $\widehat B_0$ are conducted
to compare with the direct method $\widetilde B$. We use SIR to
construct $\Kyzhat$ by categorizing $Y$ into $10$ slices. For
$\widehat B$, we use SIR to construct $\Kwzhat$ by categorizing each
component of $W$ into $2$ and $3$ slices for the cases of $n=150$
and $n=250$, respectively, and use PSIR to construct $\Kyzwhat$ by
further categorizing $Y$ into 3 slices within each slice of $W$. For
$\widehat B_0$, SIR is used to construct $\Kywzhat$ using the same
slicing as $\widehat B$. We also try other settings of slices, which
give similar results and thus are not reported. For the two-stage
method, we use $C_{n}=n^{1/4}$ and
$\Xi=\{\frac{5}{50},\frac{6}{30},\ldots,\frac{45}{50}\}$ to
determine $(\nu^*,\xi^*)$. The {\it trace correlation coefficient}
(Hooper, 1959) $r=\sqrt{\tr(P_{B_1}P_{B_2})/d}$ of two
$d$-dimensional subspaces with bases $B_1$ and $B_2$ is used as the
performance measure. The value of $r$ belongs to $[0,1]$, and $r=1$
indicates $\rmspan(B_1)=\rmspan(B_2)$. Simulation results are
reported under $\sigma=0.5$ and different combinations of
$a=(0,0.5,\ldots,3)$ and $b=(0.1,0.3)$, based on $m=300$ bootstraps
and 500 replicates.



\subsection{Simulation results}

We first compare the performances of the two-stage method $\widehat
B$ and the direct method $\widetilde B$ when $d$ is known.
Simulation results under models~(\ref{model1}) and (\ref{model2})
are placed in Figures~\ref{fig:11}-\ref{fig:12}~(a)-(b), which show
the means of the trace correlation $\widehat r$ and $\widetilde r$
of $\widehat B$ and $\widetilde B$, respectively. Recall that $a$
controls the influence of $W$ on $Y$, and $b$ controls the
correlation between $(Z,W)$. Thus, in the absence of $W$,
$\widetilde r$ decreases as $a$ increases. For any fixed $a$, we can
also observe a high $\widetilde r$ for large $b$, since in this
situation $Z$ can well predict $Y$ through $W$. A similar pattern
can be observed for $\widehat r$, since $\widehat B$ is obtained
from modifying $\Kyzhat$.

The magnitude of improvement $(\widehat r-\widetilde r)$ shows a
different behavior. One can see that $(\widehat r-\widetilde r)$
increases as $a$ increases. The more information $W$ contains (i.e.,
large $a$), the more improvement $\widehat B$ can achieve. On the
other hand, $(\widehat r-\widetilde r)$ increases as $b$ decreases.
Note that when $b$ is small, $Z$ can hardly be a surrogate of $W$
and, hence, the two-stage method benefits more from utilizing $W$.
Overall, $\widehat B$ outperforms $\widetilde B$ in all settings.

The trace correlation $\widehat r_0$ of $\widehat B_0$ is also shown
in Figures~\ref{fig:11}-\ref{fig:12}~(a)-(b). One can obviously see
that the winner is still $\widehat B$, followed by $\widehat B_0$
and then $\widetilde B$. These results further indicate the benefit
of using a hybrid method to estimate $\se$ via $\Kyzw$ and $\Kwz$.
Indeed, we rarely know the relative importance of $\Kyzw$ and
$\Kwz$, which is totally ignored when directly estimating $\se$ by
$\Kywzhat$. By using the hybrid kernel matrix $\widehat K(\xi)$, it
allows the data to select $\xi$ with minimal variability to adapt to
various relationships between $\Kyzw$ and $\Kwz$, and a good
performance of $\widehat B$ is achieved.

The simulation results of $\widehat d$ using $\Kehat=\widehat
K(\xi)$ and different $C_n$ values are placed in
Table~\ref{table:1}, which shows the selection proportions under
four settings of two models with $(n,p)=(150,9)$. The results of
$\widetilde d$, the estimator of $d$ from using $\Kyzhat$ directly,
are also shown for comparisons. One can see that $\widehat d$
achieves higher accuracies than $\widetilde d$ over a wide range of
$C_n$. By incorporating $W$, we cannot only improve estimating
$\syz$, but also improve the selection consistency for the
structural dimension $d$.




\section{Data Analysis}

\subsection{The Pima Indians diabetes data} \label{sec.data.pima}

The data contains females of Pima Indian heritage, each with 7
biological covariates, 1 covariate of family disease history, and an
indicator of diabetes status. Detailed description of the data can
be found in Smith \textit{et al.} (1988). After removing
observations with missing values, there are 392 patients remained.
In our analysis, we take the score of the family disease history as
$W$, which is shown to have strong association with diabetes status
$(Y)$. However, missingness is very likely to occur when collecting
$W$, which limits its usage in future applications. It is thus of
interest to construct prediction rule for the diabetes status based
solely on the rest biological measurements $(X)$.

Since SIR can only find one direction for binary response, we use
SAVE to construct $\Kyzhat$. As to the estimation of $\se$, we apply
PSIR to construct $\Kyzwhat$, and apply SIR to construct $\Kwzhat$.
In this analysis, we choose $(\nu,\xi)$ of the two-stage method via
maximizing the leave-one-out classification accuracy (CA) from
quadratic discriminant analysis (Rencher, 1995), which gives
$\widehat B=\widehat B(4,0.2)$.
The maximum leave-one-out CA is $0.7959$ for $(\widehat
S_{1},\widehat S_{2})$, while it is $0.7041$ for $(\widetilde
S_1,\widetilde S_2)$, and is $0.7781$ for $(\widetilde
S_1,\widetilde S_2,\widetilde S_3)$. It indicates an efficiency gain
from using $W$. This fact can be further observed from the scatter
plots of $\widehat S_i$ and $\widetilde S_i$ in Figure~\ref{fig:2}.
One can see that $(\widetilde S_{1}, \widetilde S_{2})$ tend to
separate diabetes status by variation, while different locations of
two groups are detected for $\widetilde S_3$. Interestingly,
$(\widehat S_1,\widehat S_2,\widehat S_3)$ demonstrate different
behaviors in that patients with different diabetes status tend to
have different variations of $\widehat S_1$, while different
locations are observed for $\widehat S_2$. Moreover, $\widehat S_3$
is useless in separating the diabetes status, suggesting that one
only requires $(\widehat S_1,\widehat S_2)$ to infer the diabetes
status. However, it requires $(\widetilde S_1,\widetilde
S_2,\widetilde S_3)$ when ignoring $W$. By utilizing $W$, the order
of the ``location-separating component'' is also changed from
$\widetilde S_3$ to $\widehat S_2$.

Following the procedure of Li (2006), we also implement the
quadratic discriminant analysis to classify subjects by using $W$
together with the
leading two components of $\Kyzwhat$. 
The result is treated as the benchmark since $W$ is directly used in
the classification process. The resulting leave-one-out CA for the
benchmark method is $0.7985$. By using the two-stage procedure, we
only require two dimension reduction components of biological
measurements to infer the behavior of diabetes status, and can
achieve comparable performance as the benchmark method.


\section{Discussion}

In this article, we propose a general framework to utilize the
additional information of $W$ to improve estimating $\syz$, via
constructing the $W$-envelope subspace $\se$.

The central mean subspace $\syxe$ is the minimum subspace
$\rmspan(\Gamma_m)$ such that $E[Y|X]=E[Y|\Gamma_m^TX]$. Obviously,
we must have $\syxe \subseteq \syx$. Therefore, $\syx$ may contain
redundant directions to infer $E[Y|X]$, and $\syxe$ should be the
target when the research interest is the conditional mean $E[Y|X]$.
The idea of $\syxe$ is first proposed by Cook and Li (2002), wherein
the estimation method is also developed. Later, Xia {\it et al.}
(2002) propose the minimum average variance estimation (MAVE), which
is based on local linear smoothers and does not require strong
assumptions on the distribution of $X$. With the presence of $W$, an
interesting question is how to utilize $W$ to improve the estimation
of $\syxe$. By definition, we have the inclusion property
\begin{eqnarray*}
\syze \subseteq \syz \subseteq \se.
\end{eqnarray*}
It implies that the idea of the envelope subspace $\se$ can still be
applied, to confine the searching space of $\syxe$ and, hence, to
enhance the estimation efficiency. Although the idea is
straightforward, efforts should be made to adapt to different
estimation criteria (such as MAVE) for $\syxe$.\\


\begin{center}
{\Large \textbf{References}}
\end{center}

\begin{description}

\item
Chiaromonte, F., Cook, R. D., and Li, B. (2002). Sufficient dimension reduction in regressions with categorical predictors. {\it
The Annals of Statistics}, 30, 475-497.
\item
Cook, R. D., and Weisberg, S. (1991). Discussion of ``Sliced inverse regression for dimension reduction,'' By K.-C. Li. {\it Journal of the American Statistical
Association}, 86, 328-332.

\item
Cook, R. D., and Yin, X. (2001). Dimension eduction and
visualization in discriminant analysis (with discussion). {\it Aust.
N. Z. J. Stat.}, 43, 147-199.

\item
Cook, R. D. (1998). \emph{Regression Graphics: Ideas for Studying
Regressions through Graphics}. New York: Wiley.
\item
Cook, R. D., and Li, B. (2002). Dimension reduction for the conditional mean in regression. {\it The Annals of Statistics}, 30, 455-474.
\item
Ding, X., and Wang, Q. (2011). Fusion-refinement procedure for dimension reduction with missing response at random. {\it Journal of the American Statistical Association}, 106, 1193-1207.
\item
Hooper, J. (1959). Simultaneous equations and canonical correlation theory. {\it Econometrica}, 27, 245-256.


\item
Hung, H. (2012). A two-stage dimension reduction method for
transformed response and its applications. {\it Biometrika}, 99,
865-877.

\item
Robins, J. M. and Greenland, S. (1992). Identifiability and
exchangeability for direct and indirect effects. {\it Epidemiology},
3, 143-155.

\item
Li, B., Wen, S., and Zhu, L. (2008). On a projective resampling method for dimension reduction with multivariate responses. {\it Journal of the American Statistical Association}, 103, 1177-1186.
\item
Li, K. C. (1991). Sliced inverse regression for dimension reduction
(with discussion). {\it Journal of the American Statistical
Association}, 86, 316-327.
\item
Li, K. C., Wang, J. L., and Chen, C. (1999). Dimension reduction for
censored regression data. {\it The Annals of Statistics}, 27, 1-23.
\item
Li, L. (2006). Survival prediction of diffuse large-B-cell lymphoma
based on both clinical and gene expression information. {\it
Bioinformatics}, 22, 466-471.




\item
Mackinnon, D. P., Fairchild, A. J., and Fritz, M. S. (2007).
Mediation analysis. {\it Annu Rev Psychol.}, 58: 593.

\item
Pearl J. (2001). Direct and indirect effects. {\it Proceedings of the Seventh Conference on Uncertainty in Artificial Intelligence}, 411-420.

\item
Rencher, A. (2002). \emph{Methods of Multivariate Analysis}. New
York: Wiley
\item
Smith, J., Everhart, J., Dickson, W., Knowler, W., and  Johannes, R. (1988). Using the ADAP learning algorithm to forecast the onset of diabetes mellitus. {\it Proceedings of the Symposium on Computer Applications and Medical Care}, 9, 261-265.
\item
Naik, P. A., and Tsai, C. L. (2005). Constrained inverse regression for incorporating prior information. {\it Journal of the American Statistical Association}, 100, 204-211.
\item
Xia, Y., Tong, H., Li, W. K., and Zhu, L. X. (2002). An adaptive estimation of dimension reduction space. {\it Journal of the Royal Statistical Society, Ser. B}, 64, 363-410.
\item
Ye, Z., and Weiss, R. E. (2003). Using the bootstrap to select one of  a new class of dimension reduction methods. {\it Journal of the
American Statistical Association}, 98, 968-979.
\item
Zhu, L. P., and Zhu, L. X. (2009). A data-adaptive hybrid method for dimension reduction.  {\it Journal of Nonparametric Stistics}, 21,
851-861.
\item
Zhu, L. P., Zhu, L. X., Ferre, L., and Wang, T. (2010). Sufficient dimension reduction through discretization-expectation estimation.
{\it Biometrika}, 97, 295-304.
\item
Zhu, L. X., Ohtakic, M., and Lid, Y. (2007). On hybrid methods of inverse regression-based algorithms. {\it Journal of Computational
Statistics $\&$ Data Analysis}, 51, 2621-2635.

\end{description}

\newpage

\begin{table}
\centering \caption{Selection proportions of $\widehat{d}$ and
$\widetilde{d}$ under models (\ref{model1}) and (\ref{model2}) with
different combinations of $(a,b,C_{n})$. The columns correspond to
the true dimension $d$ are marked as bold.} \footnotesize
\begin{tabular}{ccc|cccc|cccc|cccc}
\hline \hline
& & & \multicolumn{4}{c|}{$C_n=0.5n^{1/4}$} & \multicolumn{4}{c|}{$C_n=n^{1/4}$} & \multicolumn{4}{c}{$C_n=2n^{1/4}$} \\
\cline{4-15}
model & $(a,b)$& & 1 & 2 & 3 & $> 3$ & 1 & 2 & 3 & $> 3$ & 1 & 2 & 3 & $> 3$ \\
\hline
\multirow{8}{*}{(\ref{model1})}&\multirow{2}{*}{$(1,0.1)$}
& $\widehat{d}$  & \bf .365 & .625 & .010 &  .000   & \bf .855 & .145 & .000    & .000 & \bf .985 & .015 & .000 & .000    \\
& & $\widetilde{d}$ & \bf .000 & .010 & .500 & .490 & \bf .000 & .205 & .735 & .060 & \bf .005 & .790 & .205 & .000    \\
\cline{2-15}
&\multirow{2}{*}{$(1,0.3)$}
& $\widehat{d}$  & \bf .470 & .520 & .010 & .000 & \bf .885 & .115 & .000 & .000 & \bf .995 & .005 & .000 & .000 \\
& & $\widetilde{d}$ & \bf .000 & .005 & .540 & .455 & \bf .000 & .205 & .750 & .045 & \bf .005 & .870 & .125 & .000 \\
\cline{2-15}
&\multirow{2}{*}{$(3,0.1)$}
& $\widehat{d}$  & \bf .015 & .490 & .485 & .010 & \bf .075 & .855 & .070 & .000 & \bf .480 & .520 & .000 & .000 \\
& & $\widetilde{d}$ & \bf .000 & .000 & .030 & .970 & \bf .000 & .000 & .265 & .735 & \bf .000 & .080 & .715 & .205 \\
\cline{2-15}
&\multirow{2}{*}{$(3,0.3)$}
& $\widehat{d}$  & \bf .010 & .650 & .335 & .005 & \bf .165 & .795 & .040 & .000 & \bf .660 & .340 & .000 & .000    \\
& & $\widetilde{d}$ & \bf .000 & .000 & .055 & .945 & \bf .000 & .005 & .475 & .520 & \bf .000 & .135 & .810 & .055 \\
\hline
\multirow{8}{*}{(\ref{model2})}&\multirow{2}{*}{$(1,0.1)$}
& $\widehat{d}$  & .000 & \bf .480 & .520 & .000 & .010 & \bf .905 & .085 & .000 & .110 & \bf .890 & .000 & .000 \\
& & $\widetilde{d}$ & .000 & \bf .005 & .175 & .820 & .000 & \bf .010 & .625 & .365 & .000 & \bf .210 & .765 & .025 \\
\cline{2-15}
&\multirow{2}{*}{$(1,0.3)$}
& $\widehat{d}$  & .000 & \bf .710 & .290 & .000 & .000 & \bf .985 & .015 & .000 & .105 & \bf .895 & .000 & .000 \\
& & $\widetilde{d}$ & .000 & \bf .000 & .220 & .780 & .000 & \bf .045 & .660 & .295 & .000 & \bf .370 & .615 & .015  \\
\cline{2-15}
&\multirow{2}{*}{$(3,0.1)$}
& $\widehat{d}$  & .005 & \bf .535 & .450 & .010 & .085 & \bf .845 & .070 & .000 & .545 & \bf .455 & .000 & .000 \\
& & $\widetilde{d}$ & .000 & \bf .000 & .125 & .875 & .000 & \bf .010 & .570 & .420 & .000 & \bf .260 & .715 & .025 \\
\cline{2-15}
&\multirow{2}{*}{$(3,0.3)$}
& $\widehat{d}$  & .000 & \bf .535 & .465 & .000 & .015 & \bf .895 & .090 & .000 & .185 & \bf .815 & .000 & .000 \\
& & $\widetilde{d}$ & .000 & \bf .000 & .105 & .895 & .000 & \bf .010 & .530 & .460 & .000 & \bf .190 & .775 & .035 \\
\hline \hline
\label{table:1}
\end{tabular}
\end{table}

\clearpage

\begin{figure}[!ht]
\centering
\includegraphics[width=2.5in]{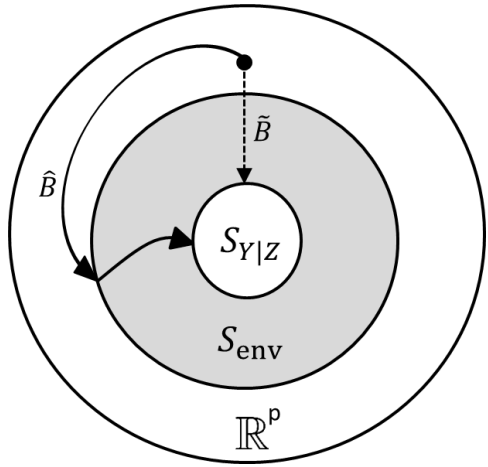}
\caption{The conceptual display of the two-stage method $\widehat B$
(solid line) and the direct method $\widetilde B$ (dashed line).}
\label{env.plot}
\end{figure}


\begin{figure}[!ht]
\centering
\includegraphics[width=6in]{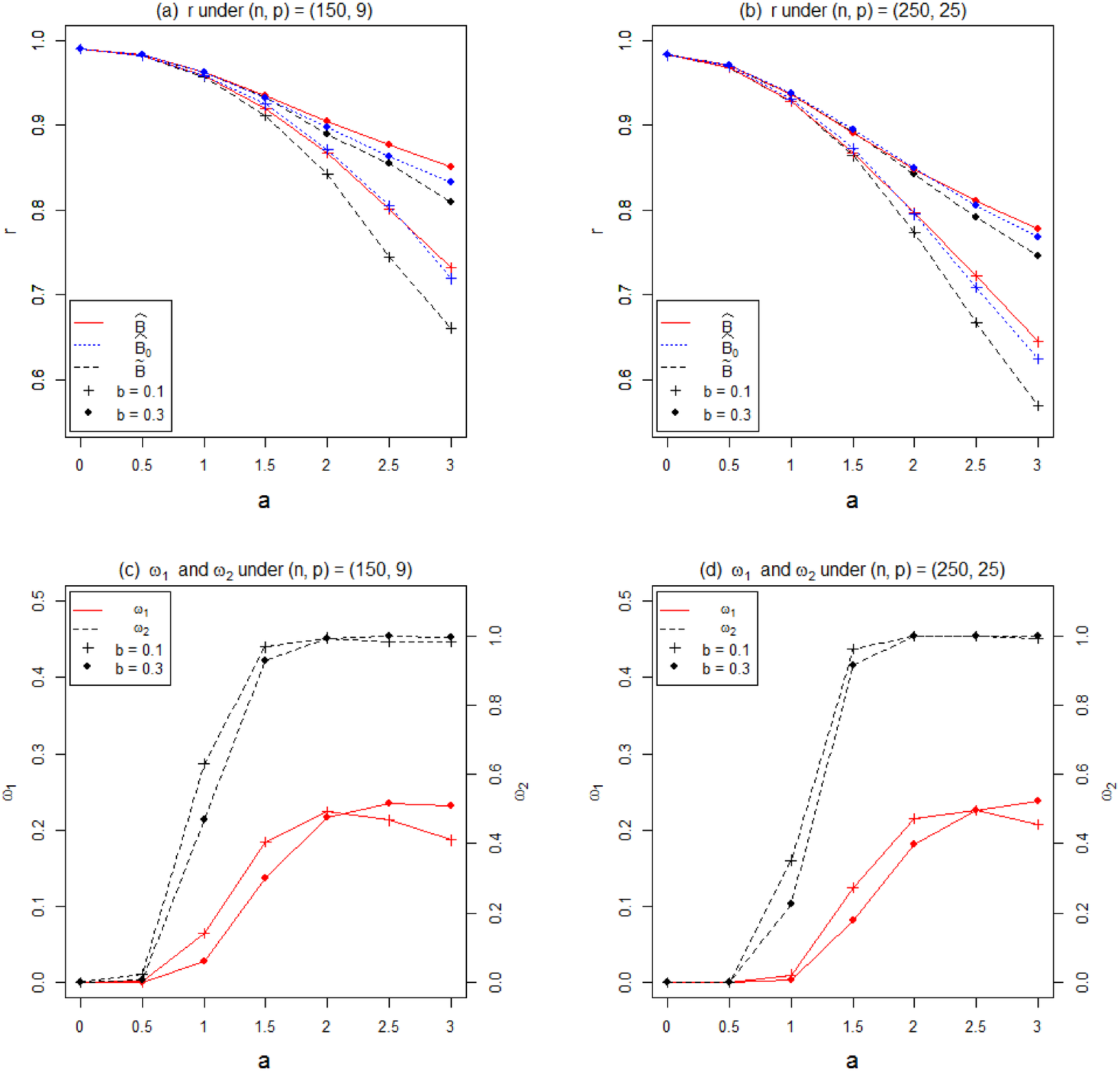}
\caption{Figures (a)-(b) are trace correlation coefficients and
Figures (c)-(d) are variabilities under model~(\ref{model1}). The
left panels show the case of $(n,p)=(150,9)$ and the right panels
show the case of $(n,p)=(250,25)$.} \label{fig:11}
\end{figure}

\begin{figure}[!ht]
\centering
\includegraphics[width=6in]{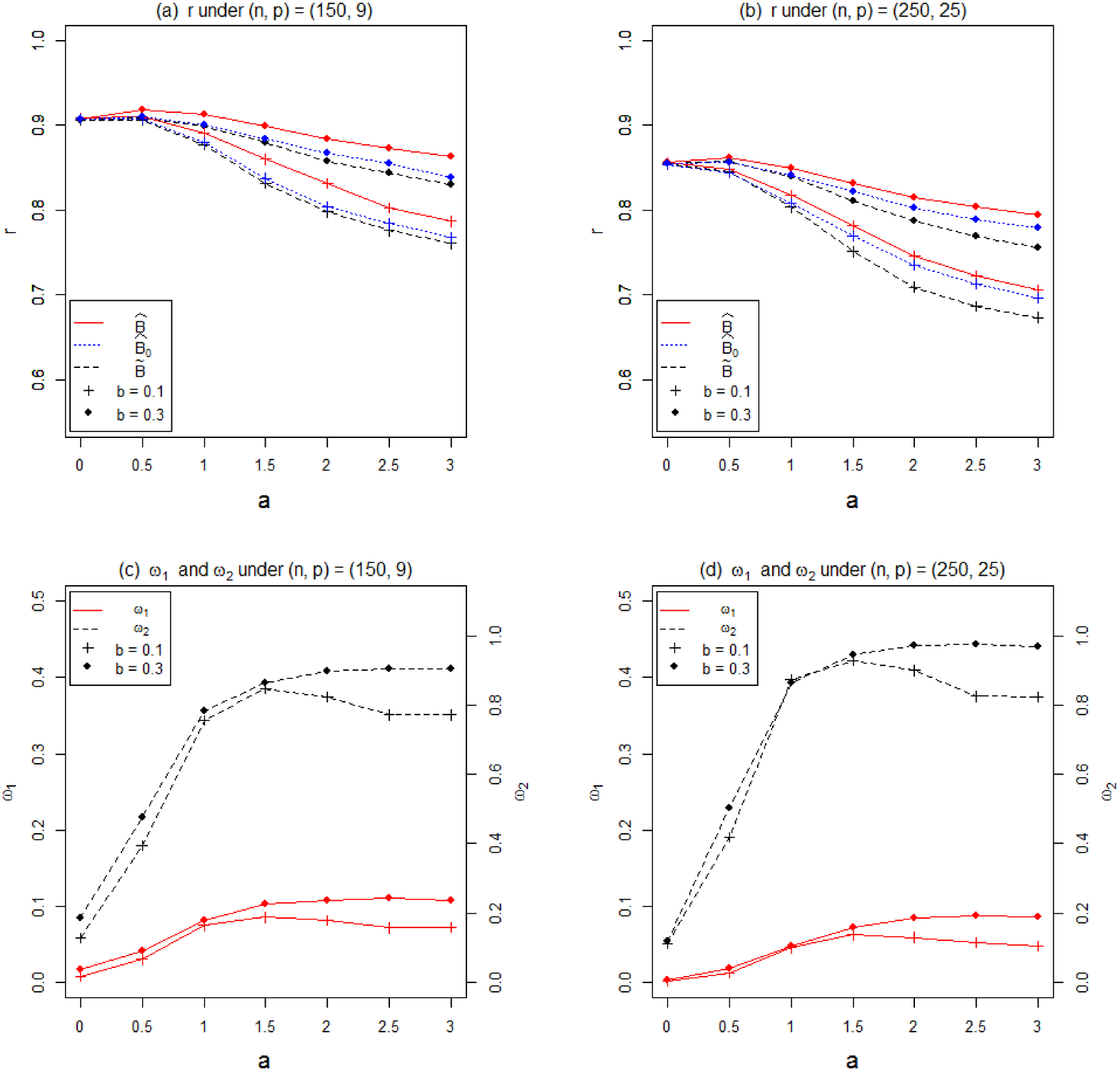}
\caption{{ Figures (a)-(b) are trace correlation coefficients and
Figures (c)-(d) are variabilities under model~(\ref{model2}). The
left panels show the case of $(n,p)=(150,9)$ and the right panels
show the case of $(n,p)=(250,25)$.}} \label{fig:12}
\end{figure}

\begin{figure}[!ht]
\centering
\includegraphics[width=6in]{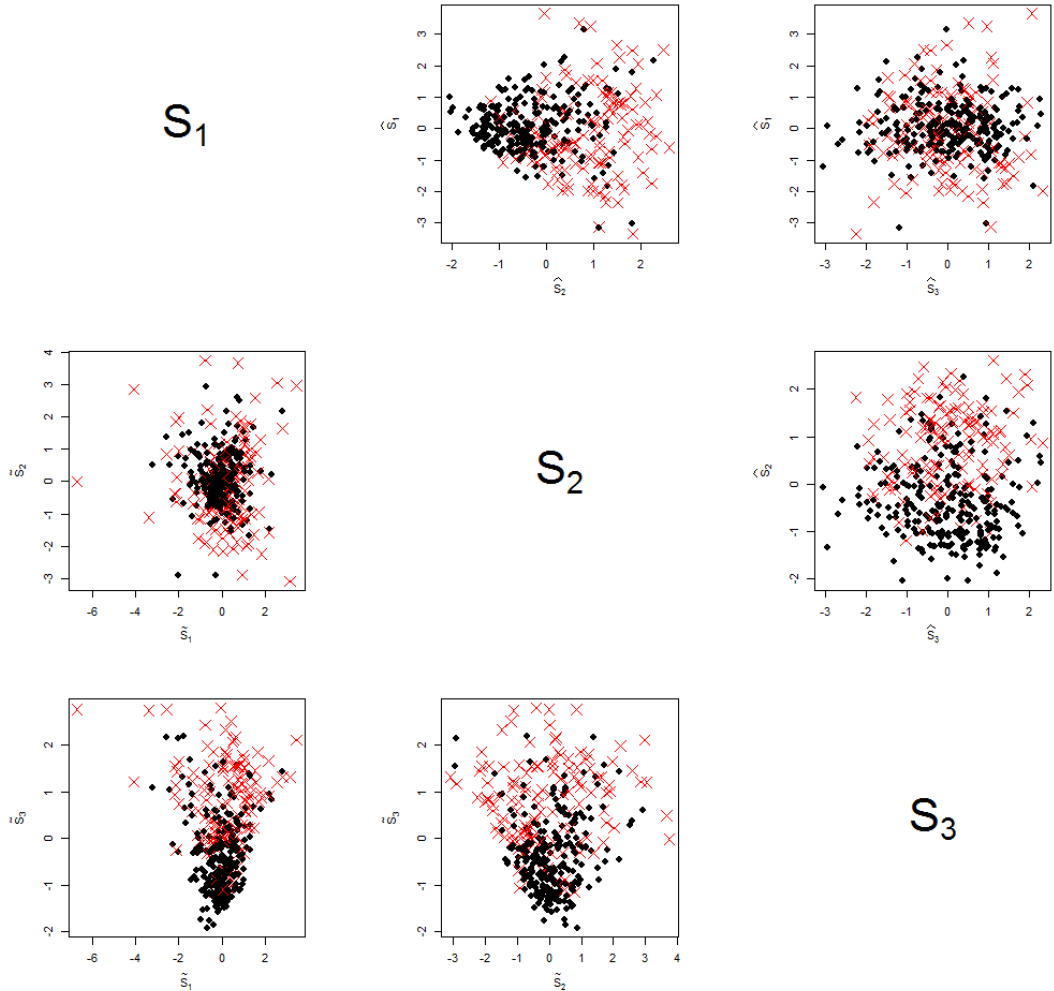}
\caption{{The scatter plot matrix of the leading three extracted
predictors from two-stage method in the upper triangular panels, and
from direct method in the lower triangular panels. $\bullet$ and
$\times$ indicate the normal and diabetes patients, respectively}}
\label{fig:2}
\end{figure}

\end{document}